\documentclass[UKenglish]{lipics_arxiv}
\title{Computing the Similarity Between Moving Curves}
\author{Kevin Buchin}
\author{Tim Ophelders}
\author{Bettina Speckmann}
\affil{Eindhoven University of Technology\\
  Eindhoven, the Netherlands\\
  \texttt{[k.a.buchin|t.a.e.ophelders|b.speckmann]@tue.nl}}
\authorrunning{K. Buchin, T. Ophelders and B. Speckmann}
\Copyright{Kevin Buchin, Tim Ophelders and Bettina Speckmann}
\subjclass{I.3.5 Computational Geometry and Object Modeling}
\keywords{Fr\'echet distance, distance between surfaces, complex moving objects, moving curves}

\usepackage{microtype}
\usepackage{color}

\usepackage{amsmath,amssymb,amsthm}
\newtheorem{thm}{Theorem}
\newtheorem{lem}[thm]{Lemma}
\newtheorem{cor}[thm]{Corollary}
\newtheorem*{rmk}{Remark}

\usepackage{graphicx}
\usepackage{wrapfig}
\usepackage{xspace}

\usepackage{pdfsync}

\newcommand{\tikzDir}{./tikz}
\newcommand{\imgDir}{./img}

\newcommand{\tIff}{if and only if\xspace}

\renewcommand{\implies}{\ensuremath\Rightarrow}
\renewcommand{\iff}{\ensuremath\Leftrightarrow}
\newcommand{\real}{\mathbb{R}}
\newcommand{\nat}{\mathbb{N}}

\newcommand{\e}{\varepsilon}
\newcommand{\F}{\mathcal F}
\newcommand{\FD}{Fr\'e\-chet distance\xspace}
\newcommand{\mtime}{\textit{time}}

\newcommand{\Rtrans}{\ensuremath{\mathrel{\text{\raisebox{2.5pt}{~\circle{6}}$\mkern-14.2mu\triangleright$}}}\xspace}
\newcommand{\inRtrans}[2]{#1\Rtrans#2}
\newcommand{\notinRtrans}[2]{\ensuremath{\neg(#1\Rtrans#2)}}

\newcommand{\R}{\ensuremath\triangleright\xspace}

\newcommand{\inR}[2]{#1\R#2}

\newcommand{\fsdII}{\ensuremath{\F^\text{2D}_\e}\xspace}
\newcommand{\fsdIII}{\ensuremath{\F_\e}\xspace}

\newcommand{\E}[1]{\ensuremath{E_{#1}^\text{d}}\xspace}

\newcommand{\matchII}{identity\xspace}
\newcommand{\matchIC}{synchronous constant\xspace}
\newcommand{\matchID}{synchronous dynamic\xspace}
\newcommand{\matchCC}{asynchronous constant\xspace}
\newcommand{\matchCD}{asynchronous dynamic\xspace}
\newcommand{\matchDD}{orientation-preserving\xspace}
\newcommand{\MatchII}{Identity\xspace}
\newcommand{\MatchIC}{Synchronous Constant\xspace}
\newcommand{\MatchID}{Synchronous Dynamic\xspace}
\newcommand{\MatchCC}{Asynchronous Constant\xspace}
\newcommand{\MatchCD}{Asynchronous Dynamic\xspace}
\newcommand{\MatchDD}{Orientation-Preserving\xspace}

\newcommand{\homII}{\includegraphics{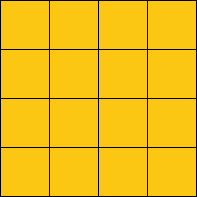}}
\newcommand{\homIC}{\includegraphics{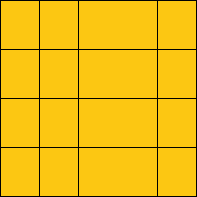}}
\newcommand{\homID}{\includegraphics{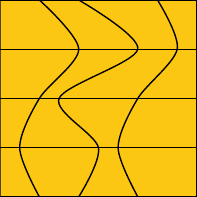}}

\newcommand{\homCC}{\includegraphics{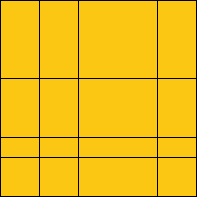}}
\newcommand{\homCD}{\includegraphics{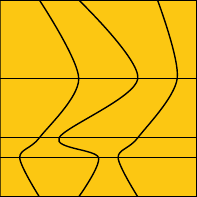}}

\newcommand{\homDD}{\includegraphics{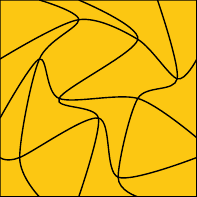}}

\newcommand{\myparNS}[1]{\noindent{\sffamily\bfseries #1.}}
\newcommand{\mypar}[1]{\medskip\myparNS{#1}}

\begin{document}
	\maketitle
	\begin{abstract}
		In this paper we study similarity measures for moving curves which can, for example, model changing coastlines or retreating glacier termini.
		Points on a moving curve have two parameters, namely the position along the curve as well as time. We therefore focus on similarity measures for surfaces, specifically the~\FD between surfaces.
		While the~\FD between surfaces is not even known to be computable, we show for variants arising in the context of moving curves that they are polynomial-time solvable or NP-complete depending on the restrictions imposed on how the moving curves are matched. We achieve the polynomial-time solutions by a novel approach for computing a surface in the so-called free-space diagram based on max-flow min-cut duality.
	\end{abstract}

\section{Introduction}

Over the past years the availability of devices that can be used to track moving objects has increased dramatically, leading to an explosive growth in movement data. Naturally the goal is not only to track objects but also to extract information from the resulting data. Consequently recent years have seen a significant increase in the development of methods extracting knowledge from moving object data.

Tracking an object gives rise to data describing its movement. Often the scale at which the tracking takes place is such that the tracked objects can be viewed as point objects. Cars driving on a highway, birds foraging for food, or humans walking in a pedestrian zone: for many analysis tasks it is sufficient to consider objects as moving points. Hence the most common data sets used in movement data processing are so-called trajectories: sequences of time-stamped points.	

However, not all moving objects can be reasonably represented as points.
A hurricane can be represented by the position of its eye, but a more accurate description is as a 2-dimensional region which represents the hurricane’s extent.
When studying shifting coastlines, reducing the coastline to a point is obviously unwanted: one is actually interested in how the whole coast line moves and changes shape over time.
The same holds true when studying the terminus of a glacier. In such cases, the moving object is best represented as a polyline rather than by a single point.
In this paper we hence go beyond the basic setting of moving point objects and study moving complex, non-point objects. Specifically, we focus on similarity measures for moving curves, based on the \FD.

\mypar{Definitions and Notation}
The \FD is a well-studied distance measure for shapes, and is commonly used to determine the similarity between two curves~$A$ and~$B:[0,1]\rightarrow\real^n$. A natural generalization 
to more complex shapes uses the definition of~Equation~\ref{eq:fdX} where the shapes~$A$ and~$B$ have type~$X\rightarrow\real^n$.
\begin{equation}\label{eq:fdX}
    D_{\text{fd}}(A,B) = \inf_{\mu:X\rightarrow X}\sup_{x\in X}\|A(x)-B(\mu(x))\|
\end{equation}

\begin{wrapfigure}{r}{0pt}%
	\includegraphics{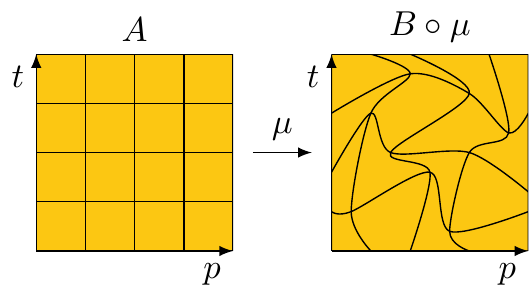}%
	\caption{A matching~$\mu$ between surfaces~$A$ and~$B$ drawn as a homeomorphism between their parameter spaces.}%
	\label{fig:hom}%
\end{wrapfigure}
Here,~$\|\cdot\|:\real^n\rightarrow \real$ is a norm such as the Euclidean norm ($L^2$) or the Manhattan norm ($L^1$). The~\emph{matching}~$\mu$ ranges over orientation-preserving homeomorphisms (possibly with additional constraints) between the parameter spaces of the shapes compared; as such, it defines a correspondence between the points of the compared shapes. A matching between surfaces with parameters~$p$ and~$t$ is illustrated in Figure~\ref{fig:hom}. Given one such matching we obtain a distance between $A$ and $B$ by taking the largest distance between any two corresponding points of~$A$ and~$B$. The \FD is the infimum of these distances taken over all possible matchings.
For moving points or static curves, we have as parameter space~$X=[0,1]$ and for moving curves or static surfaces, we have~$X=[0,1]^2$. We can define various similarity measures between shapes by imposing further restrictions on~$\mu$.

In practice a curve is generally represented by a sequence of~$P+1$ points. Assuming a linear interpolation between consecutive points, this results in a polyline with~$P$ segments. Analogously, a moving curve is a sequence of~$T+1$ polylines, each of~$P$ segments. We also interpolate the polylines linearly, yielding a bilinear interpolation, or a quadrilateral mesh of~$P\times T$ quadrilaterals.

\mypar{Related Work}
The \FD or related measures are frequently used to evaluate the similarity between
point trajectories~\cite{commuting,constrained,context}.
The \FD is also used to match point trajectories to a street network~\cite{map,mapMatching}.
The \FD between polygonal curves can be computed in near-quadratic time~\cite{fd,Bringmann14,soviets,extended}, and approximation algorithms~\cite{AronovHKWW06,driemel2012approximating}
have been studied.

The natural generalization to moving (parameterized) curves is to interpret the curves as surfaces parameterized over time and over the curve parameter. The \FD between surfaces is NP-hard~\cite{godau}, even for terrains~\cite{simpleHard}. In terms of positive algorithmic results for general surfaces the \FD is only known to be semi-computable~\cite{surface,Buchin07}. Polynomial-time algorithms have been given for the so called weak \FD~\cite{surface} and for the \FD between simple polygons~\cite{simplePoly} and so called folded polygons~\cite{foldedPoly}.

When interpreting moving curves as surfaces it is important to take the different roles of the two parameters into account: the first is inherently linked to time and the other to space. 
This naturally leads to restricted versions of the \FD of surfaces.
For curves, restricted versions of the \FD were considered~\cite{constrained,speedlimits}. For surfaces we are not aware of similar results.

\subsection{Results}
			We refine the \FD between surfaces to meaningfully compare moving curves.
			To do so, we restrict matchings to be one of several suitable classes.
			Representative matchings for the considered classes together with the running times of our results are illustrated in Figure~\ref{fig:results}.

			\begin{wrapfigure}{r}{0pt}
				\begin{tabular}{c}%
					\homII\\[-1pt]
					\MatchII\\[-1pt]
					$O(PT)$\\[1ex]

					\homIC\\[-1pt]
					\MatchIC\\[-1pt]
					$O(P^2T\log(PT))$\\[1ex]

					\homID\\[-1pt]
					\MatchID\\[-1pt]
					$O(P^3T\log P\log(PT))$\\[1ex]
					
					\homCC\\[-1pt]
					\MatchCC\\[-1pt]
					NP-complete\\[1ex]

					\homCD\\[-1pt]
					\MatchCD\\[-1pt]
					NP-hard\\[1ex]

					\homDD\\[-1pt]
					\MatchDD\\[-1pt]
					NP-hard
				\end{tabular}%
				\caption{The time complexities of the considered classes of matchings.\label{fig:results}}%
			\end{wrapfigure}

			The simplest class of matchings consists of a single predefined~\emph{\matchII} matching~$\mu(p,t)=(p,t)$.
			Hence, to compute the~\emph{\matchII} \FD, we need only determine a pair of matched points that are furthest apart.
			It turns out that one of the points of a furthest pair is a vertex of a moving curve (i.e. quadrilateral mesh), allowing computation in~$O(PT)$ time, see Section~\ref{sec:identity}.

			We discuss the~\emph{\matchIC} \FD in Section~\ref{sec:constant}.
			Here we assume that the matching of timestamps is known in advance, and the matching of positions is the same for each timestamp, so it remains constant.
			Our algorithm computes the positional matching minimizing the \FD.

			The~\emph{\matchID} \FD considered in Section~\ref{sec:dynamic} also assumes a predefined matching of timestamps, but does not have the constraint of the~\emph{\matchIC} class that the matching of positions remains constant over time.
			Instead, the positional matching may change continuously over time.

			Finally, in Section~\ref{sec:hardness}, we consider several cases where neither positional nor temporal matchings are predefined.
			The three considered cases are the~\emph{\matchCC},~\emph{\matchCD}, and~\emph{\matchDD} \FD.
			The~\matchCC class of matchings consists of a constant (but not predefined) matching of positions, as well as timestamps whereas in the~\matchCD class of matchings, the positional matching may change continuously.
			In the~\matchDD class, matchings range over orientation preserving homeomorphisms between parameter spaces, given that the corners of the parameter spaces are aligned.

			The last three classes are quite complex, and we give constructions proving that approximating the \FD within a factor~$1.5$ is NP-hard under these classes.
			For the~\matchCC and~\matchCD classes of matchings, this result holds even for moving curves embedded in~$\real^1$ whereas the result for the~\matchDD case holds for embeddings in~$\real^2$.

			Although we do not discuss classes where positional matchings are known in advance, these symmetric variants can be obtained by interchanging the time and position parameters for the discussed classes.
			Deciding which variant is appropriate for comparing two moving curves depends largely on how the data is obtained, as well as the use case for the comparison.
			For instance, the \matchIC variant may be used on a sequence of satellite images which have associated timestamps.
			The \matchID \FD is better suited for sensors with different sampling frequencies, placed on curve-like moving objects.

	\section{\MatchII Matchings}\label{sec:identity}
		\begin{wrapfigure}{l}{0pt}\homII\end{wrapfigure}
		Suppose we are given a single predefined matching~$\mu$ between the moving curves~$A$ and~$B:[0,P]\times[0,T]\rightarrow\real^n$.
		We can compute the \FD under this matching if we can find the corresponding points of~$A$ and~$B$ that are furthest apart, see Equation~\ref{eq:fdMu}.
		\begin{equation}\label{eq:fdMu}
		D_{\mu}(A,B) = \sup_{x\in [0,P]\times[0,T]}\|A(x)-B(\mu(x))\|
		\end{equation}

		If~$A$ and~$B$ are both of size~$P\times T$, then the identity matching~$\mu(p,t)=(p,t)$ allows us to simplify~$A(x)-B(\mu(x))$ into~$C(x)=A(x)-B(x)$ where~$C$ is again a quadrilateral mesh of size~$P\times T$.
		The \FD for the identity matching depends only on the point on~$C$ that is furthest from the origin.
		To see this, consider a single quadrilateral, the point furthest from the origin must be one its four corners since all points on the quadrilateral lie within the convex hull of its four corner points.
		Hence, it suffices to check only the distance to the origin for the~$O(PT)$ vertices of~$C$, see Equation~\ref{eq:fdC}.
		The \FD under the identity matching can then be computed in~$O(PT)$ time.

		\begin{equation}\label{eq:fdC}
		D_{id}(A,B) = \sup_{x\in [0,P]\times[0,T]}\|\underbrace{A(x)-B(x)}_{C(x)}\| = \sup_{x\in \{0,\dots,P\}\times\{0,\dots,T\}}\|C(x)\|
		\end{equation}

		Meshes of different size can be compared after introducing~$O(PT)$ dummy vertices, such that the meshes have equal dimensions and each vertex has an aligned vertex on the other mesh.
		For this, it is important to note that any quadrilateral can be subdivided into an equivalent mesh of four quadrilaterals with any point on the original quadrilateral as their shared corner.

		To extend this further, consider the case where the matching~$\mu(p,t)=(\pi(p),\tau(t))$ is defined by two piecewise-linear functions~$\pi$ and~$\tau$ of~$|\pi|$ and~$|\tau|$ vertices.
		This allows comparing moving curves under predefined realigments of timestamps as well as positions.
		For such a matching, the surface~$C(x)=A(x)-B(\mu(x))$ is a quadrilateral mesh of~$O((P+|\pi|)(T+|\tau|))$ vertices.
		We illustrate this in~Figure~\ref{fig:pwl} for~$t=4$ and~$|\tau|=3$ with vertices~$(\tau(0),\tau(1.5),\tau(4))=(0,2.5,4)$.
		In such case~$O(T+|\tau|)$ timestamps of~$A$ are matched with~$O(T+|\tau|)$ timestamps of~$B$.

		\begin{figure}[h]%
			\centering%
			\includegraphics{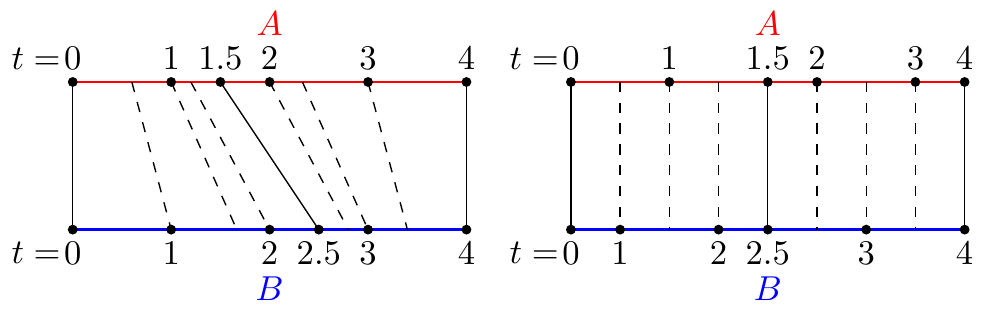}%
			\caption{A piecewise linear matching of timestamps.
			Left, the matching drawn with the original time axes.
			Right, the matching drawn after warping the time axes.}%
			\label{fig:pwl}%
		\end{figure}%

		The same can be done given a piecewise-linear reparameterization of positions.
		As a result, the \FD under a predefined piecewise-linear reparameterization of timestamps and positions can be computed in~$O((P+|\pi|)(T+|\tau|))$ time.

	\section{\MatchIC Matchings}\label{sec:constant}
		\begin{wrapfigure}{l}{0pt}\homIC\end{wrapfigure}
		In this section, we consider the class of \emph{\matchIC matchings} where the matching~$\mu(p,t)=(\pi(p),t)$ assumes no realignments of time and a constant reparameterization%
		\footnote{We make slight abuse of notation:~we allow $\mu$ to be drawn from the closure of the space of matchings. As a consequence,~$\pi(x)$ is not an actual function since some~$x$ can be matched to multiple (but connected)~$y$ values.}~$\pi$ of positions.
		Thus, the goal is to find a continuous nondecreasing surjection~$\pi$ on~$[0,P]$, such that the \FD is minimized.

		Before we present the algorithm to find \matchIC matchings, we refer to an existing algorithm~\cite{fd} that computes the \FD between static curves.
		This classic algorithm makes use of a data structure called the $\e$-freespace diagram, which is the set of parameter pairs for which the represented points are at most~$\e$ apart.
		A freespace diagram for two static curves is illustrated in Figure~\ref{fig:fsdStatic}.

		The matching~$\pi(x)=y$ between the parameter spaces of static curves~$A$ and~$B$ can be embedded as a bimonotone path~$\{(x,y)\mid \pi(x)=y\}$ in the freespace.
		Therefore, the \FD is at most~$\e$ if and only if there exists a bimonotone path from the bottom-left to the top-right corner of the $\e$-freespace.
		Using freespace diagrams, the decision problem for static curves can be solved in~$O(P^2)$ time.

		\begin{figure}[h]\centering
			\includegraphics{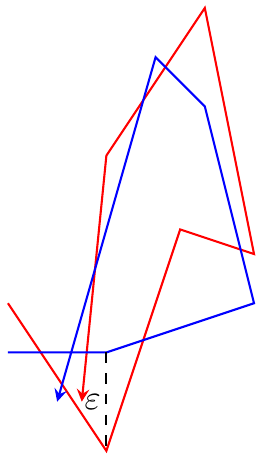}\qquad\qquad\includegraphics{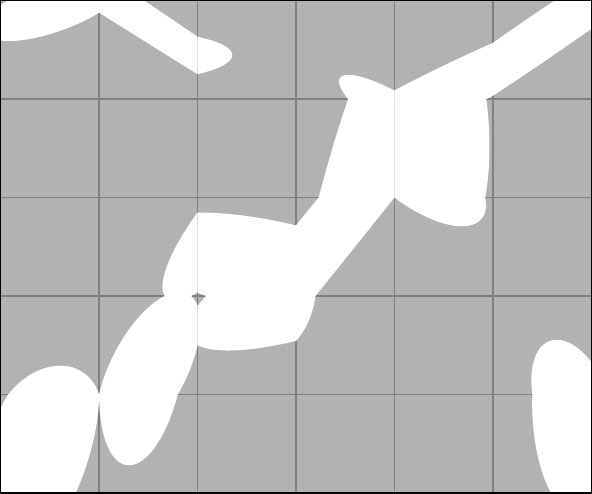}
			\caption{\label{fig:fsdStatic}
				Left: Curves $A$ (red) and $B$ (blue) with~\FD~$\e$.
				Right: In white, their freespace diagram $\F_\e=\{(x,y)\in[0,6]\times[0,5]\mid \|A(x)-B(y)\|\leq\e\}$.
				If we draw the freespace diagram for any smaller value of~$\e$, no bimonotone path through the freespace connects the bottom-left corner to the top-right corner.}
		\end{figure}

		We extend this approach to compute matchings between moving curves.
		For this, consider the 3D freespace defined by Equation~\ref{eq:fsd3d}.
		Since any \matchIC matching~$\mu(p,t)=(\pi(p),t)$ is defined by a bimonotone path~$\pi:[0,P]\rightarrow[0,P]$, the matching~$\mu$ is embedded in the freespace diagram as the surface~$\mu=\{(x,y)\mid\pi(x)=y\}\times[0,T]$.
		Such a matching yields a~\FD of at most~$\e$ if and only if~$\mu\subseteq\fsdIII$.
		The 3D freespace consists of~$O(P^2T)$ cells $C_{x,y,t}=\fsdIII\cap([x,x+1]\times[y,y+1]\times[t,t+1])$ for~$(x,y,t)\in\nat^3$.
		\begin{equation}\label{eq:fsd3d}
		\fsdIII=\{(x,y,t)\in[0,P]\times[0,P]\times[0,T]\mid \|A(x,t)-B(y,t)\|\leq\e\}
		\end{equation}

		We can simplify the three-dimensional freespace~(\fsdIII) into a two-dimensional one~(\fsdII) with
		\begin{align}
		\label{eq:fsdIsct}
		\fsdII = &
			~~\bigcap\limits_{t\in[0,T]}~~
			\{(x,y)\in[0,P]\times[0,P]\mid \|A(x,t)-B(y,t)\|\leq\e\}
		\end{align}
		and find the path~$\pi$ defining~$\mu$ in it. To simplify Equation~\ref{eq:fsdIsct} we prove the following lemma.

		\begin{lem}A cell~$C_{x,y,t}$ of the freespace has a convex intersection with any line parallel to the~$xy$-plane or the~$t$-axis.\label{thm:convex}\end{lem}
		\begin{proof}
		Each cell in the freespace is the freespace induced by a pair of quadrilaterals, so consider two quadrilaterals~$A(x,t)$ and~$B(y,t)$.
		These quadrilaterals are bilinear interpolations between four corner points.
		Hence,~$A(x,t)-B(y,t)$ is an affine map for each fixed~$t$.
		Likewise,~$A(x,t)-B(y,t)$ is an affine map for each fixed pair~$(x,y)$.
		Because the preimage of a convex norm ball under an affine map is convex, the intersection of a freespace cell with lines parallel to the~$t$-axis or the~$xy$-plane forms a convex set.
		\end{proof}

		\begin{lem}\label{thm:convexSnap}
		$\fsdII =
			\bigcap\limits_{t\in\{0,\dots,T\}}
			\{(x,y)\in[0,P]\times[0,P]\mid \|A(x,t)-B(y,t)\|\leq\e\}$.
		\end{lem}
		\begin{proof}
		By Lemma~\ref{thm:convex}, the intersection of a cell of~$\fsdIII$ with a line parallel to the~$t$-axis is convex.
		Hence, if a point~$(x,y,t)\notin\fsdIII$, then either~$(x,y,\lfloor t\rfloor)\notin\fsdIII$ or~$(x,y,\lceil t\rceil)\notin\fsdIII$.
		Therefore the internal structure of freespace cells can safely be ignored.
		\end{proof}

		We denote the~$O(P^2)$ cells of the 2D freespace by~$C_{x,y}=([x,x+1]\times[y,y+1])\cap\fsdII$ where~$(x,y)\in\nat^2$. 
		\begin{lem}\label{thm:isctConvexCell}
		Every cell~$C_{x,y}$ in the 2D freespace~$\fsdII$ is convex.
		\end{lem}
		\begin{proof}
		By Lemma~\ref{thm:convex}, each cell is the intersection of convex sets, so each~$C_{x,y}$ is convex.
		\end{proof}

		\begin{wrapfigure}[13]{r}{0pt}%
			\includegraphics{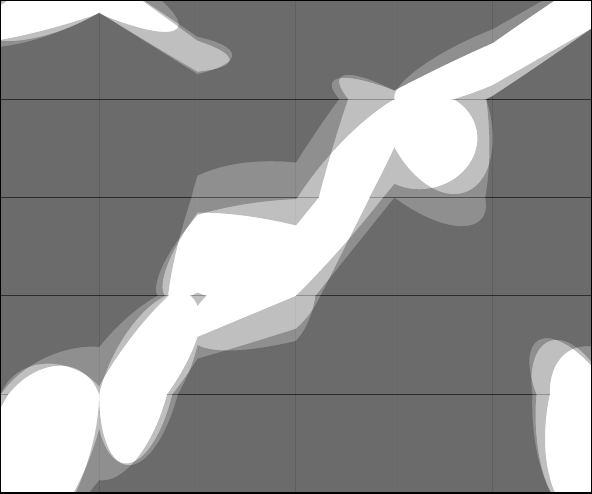}
			\caption{\fsdII for two moving curves. Darker areas are in fewer layers of~\fsdIII.\label{fig:fsdII}}
		\end{wrapfigure}
		Figure~\ref{fig:fsdII} illustrates (in white) what the 2D freespace~$\fsdII$ might look like for two moving curves.
		As before, any smaller value of~$\e$ would disconnect the bottom-left from the top-right corner.
		For a given~$\e$, for each of the~$O(P^2)$ cells, the boundary intervals can be computed in~$O(T)$ time.
		Therefore, finding an $xy$-monotone path takes~$O(P^2T)$ time, solving the decision problem for the \FD under \matchIC matchings in~$O(P^2T)$ time.

		\subsection{Parametric Search}\label{sec:paramConst}
			To compute the exact \FD efficiently, we use the decision problem in a parametric search for the minimum~$\e$ admitting a matching.
			When increasing~$\e$ starting from~$\e=0$, there are three types of critical values of~$\e$ for which a passage might open in the freespace:

			\begin{enumerate}
			\item[a)]The minimal~$\e$ with~$(0,0)\in\fsdII$ and~$(P,P)\in\fsdII$.
			\item[b)]One of the four boundaries of a cell~$C_{x,y}$ becomes nonempty.
			\item[c)]The lower (or left) endpoint on the boundary of cell~$C_{x,y}$ aligns with the upper (or right) endpoint on the boundary of~$C_{x+i,y}$ (or~$C_{x,y+j}$).
			\end{enumerate}

			Note that all critical values occur when two endpoints align (or an endpoint aligns with a gridpoint).
			For each of the~$O(T)$ layers defining~\fsdII, we use a function in~$\e$ to represent the~$x$- or~$y$-position of such endpoint (or gridpoint).
			This amounts to a total of~$O(P^2T)$ functions, and each critical value occurs when two of them intersect.
			Since the intersections between any pair of these functions can be computed in constant time, we can apply a parametric search~\cite{megiddo1983applying}.

			We use Cole's variant~\cite{cole1987slowing} of the parametric search to find the desired critical value in~$O((k+\mtime_\text{dec})\log k)$ time.
			Here~$k=O(P^2T)$ is the number of functions to which the parametric search is applied and~$\mtime_\text{dec}=O(P^2T)$ is the running time for the decision problem.
			We obtain the running time of Theorem~\ref{thm:fdIC}.
			\begin{thm}\label{thm:fdIC}
			The \matchIC \FD between quadrilateral meshes can be computed in~$O(P^2T\log(PT))$ time.
			\end{thm}

			\begin{rmk}
				To compare quadrilateral meshes under a piecewise linear realignment timestamps, we can subdivide the quadrilateral meshes as explained in~Section~\ref{sec:identity}.
				Although for simplicity we have assumed that the meshes are of equal size, when computing the \FD between meshes of different sizes,~$P\times T$ and~$Q\times T$, the 2D freespace has only~$O(PQ)$ cells.
				Thus the decision problem is solved in~$O(PQT)$ time and the exact computation takes~$O(PQT\log(PQT))$ time.
			\end{rmk}

			\begin{rmk}
				If the inputs of the algorithm are two sequences of~$O(T)$ curves without predefined interpolations to obtain quadrilateral meshes, we can still measure their \FD for the optimal (but unknown) linear interpolation.
				Due to the convexity of freespace cells, the \FD is the minimum \FD between two curves~$A_t$ and~$B_t$, which can be computed in~$O(TP^2\log P)$ time by running the original algorithm~$O(T)$ times.
			\end{rmk}

\section{\MatchID Matchings}\label{sec:dynamic}
		\begin{wrapfigure}{l}{0pt}\homID\end{wrapfigure}
\emph{Synchronous dynamic} matchings align timestamps under the identity matching, but the matching of positions may change continuously over time.
		Specifically, the matching is defined as~$\mu(p,t)=(\pi_t(p),t)$.
		Here,~$\mu(p,t):[0,P]\times[0,T]\rightarrow[0,P]\times[0,T]$ is continuous, and for any~$t$ the matching~$\pi_t:[0,P]\rightarrow[0,P]$ between the two curves at that time is a nondecreasing surjection.

\subsection{Freespace Partitions in 2D}
		Recall that the freespace diagram~$\F_\e$ is the set pairs of points that are within distance~$\e$ of each other.
		$$(x,y)\in \F_\e \iff \|A(x)-B(y)\|\leq\e$$
		If~$A$ and~$B$ are curves with parameter space~$[0,P]$, then their freespace diagram is two-dimensional, and the~\FD is the minimum value of~$\e$ for which an $xy$-monotone path (representing~$\mu$) from~$(0,0)$ to~$(P,P)$ through the freespace exists.

		We use a variant of the max-flow min-cut duality to determine whether a matching through the freespace exists.
		Before we present the 3D variant for moving curves with synchronized timestamps, we illustrate the idea in the fictional 2D freespace of Figure~\ref{fig:muSep2D}.
\newpage
	\begin{wrapfigure}{r}{0pt}%
		\includegraphics{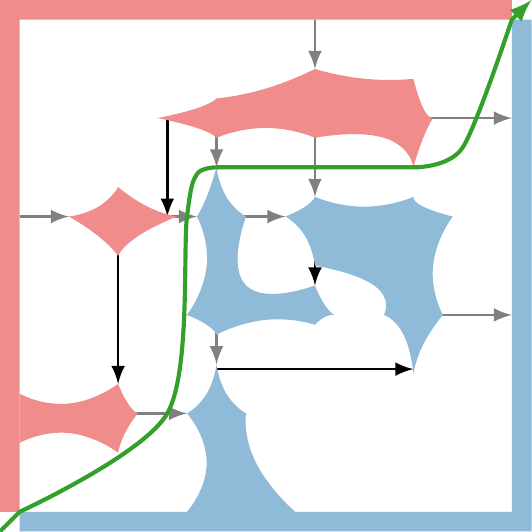}%
		\caption{A matching (green) in the 2D freespace (white).}%
		\label{fig:muSep2D}%
	\end{wrapfigure}%
		Here, any matching---such as the green path---must be an~$x$- and~$y$-monotone path from the bottom left to the top right corner and this matching must avoid all obstacles (i.e. all points not in~$\F_e$).
		Therefore each such matching divides the obstacles in two sets: those above, and those below the matching.

		Suppose we now draw a directed edge from an obstacle~$a$ to an obstacle~$b$ \tIff any matching that goes over~$a$ must necessarily go over~$b$.
		The key observation is that a matching exists unless such edges can form a path from the lower-right boundary to the upper-left boundary of the freespace.
		In the example, a few trivial edges are drawn in black and gray.
		If all obstacles were slightly larger, an edge could connect a blue obstacle with a red obstacle, connecting the two boundaries by the edges drawn in black.

	\subsection{Freespace Partitions in 3D}\label{sec:partition3D}
		In contrast to the 2D freespace where the matching is a path, matchings of the form~$\mu(p,t)=(\pi_t(p),t)$ form surfaces in the 3D freespace~\fsdIII (see Equation~\ref{eq:fsd3dMain}).
		Such a surface again divides the obstacles in the freespace in two sets and can be punctured by a path connecting two boundaries.
		We formalize this concept for the 3D freespace and give an algorithm for deciding the existence of a matching.
		\begin{equation}\label{eq:fsd3dMain}
		(x,y,t)\in\fsdIII \text{ \tIff~} \|A(x,t)-B(y,t)\|\leq\e
		\end{equation}

		For $x,y,t\in\nat$, the cell~$C_{x,y,t}$ of the 3D freespace is the set~$\fsdIII\cap([x,x+1]\times[y,y+1]\times[t,t+1])$. The property of Lemma~\ref{thm:convex} holds for all such cells.

		\begin{wrapfigure}[11]{r}{0pt}%
			\includegraphics{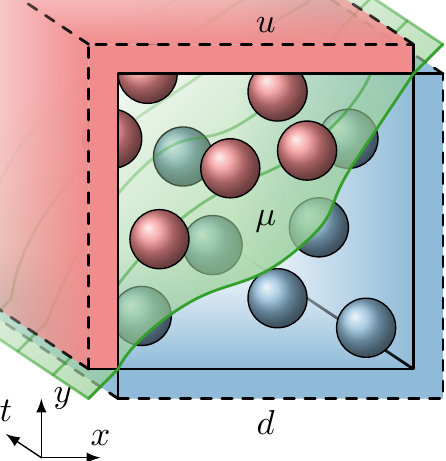}
			\caption{$\mu$ separates~$u$ and~$d$.}%
			\label{fig:muSepUD}%
		\end{wrapfigure}%
		
		We divide the set of points not in~$\fsdIII$ into a set~$O$ of so-called obstacles, such that each individual obstacle is a connected point set.
		Let~$u$ be the open set of points representing the left and top boundary of~$\fsdIII$.
		Symmetrically, let~$d$ represent the bottom and right boundary, see Figure~\ref{fig:muSepUD}.
		Denote by~$O'\subset O$ the obstacles between the boundaries.

		\begin{description}
			\item $O = \{u,d\}\cup O'$ with~$\bigcup O' = ([0,P]^2\times[0,T])\setminus \fsdIII$;
			\item $u = \{ (x,y,t) \mid (x<0 \land y>0) \lor (x<P \land y>P) \}$;
			\item $d = \{ (x,y,t) \mid (x>0 \land y<0) \lor (x>P \land y<P) \}$.
		\end{description}

		Given a matching~$\mu$, let~$D \subseteq O$ be the set of obstacles below it, then~$u\notin D$ and~$d\in D$.
		Here, we use axes~$(x,y,t)$ and say that a point is below some other point if it has a smaller~$y$-coordinate.
		Because each obstacle is a connected set and~$\mu$ cannot intersect obstacles, a single obstacle cannot lie on both sides of the same matching.
		Because all matchings have~$u \notin D$ and~$d \in D$, a matching exists \tIff $\neg(d \in D \implies u \in D)$.

		We compute a relation~\R of elementary dependencies between obstacles, such that its transitive closure~\Rtrans has~$\inRtrans{d}{u}$ if and only if $d\in D\implies u\in D$.
		Let~$\inR{a}{b}$ \tIff $a\cup b$ is connected ($a$ touches~$b$) or there exists some point~$(x_a,y_a,t_a)\in a$ and~$(x_b,y_b,t_b)\in b$ with~$x_a\leq x_b$,~$y_a\geq y_b$ and~$t_a=t_b$.
		We prove in Lemmas~\ref{lem:witnessIfNoD} and~\ref{lem:matchIfD} that this choice of~$\R$ satisfies the required properties and in Theorem~\ref{thm:fdrel} that we can use the transitive closure~$\Rtrans$ of~$\R$ to solve the decision problem of the~\FD.

		\begin{lem}\label{lem:witnessIfNoD}
			If~$\inRtrans{a}{b}$, then~$a\in D \implies b\in D$.
		\end{lem}
		\begin{proof}
			Assume that~$\inR{a}{b}$, then either~$a$ touches~$b$ and no matching can separate them, or there exists some~$(x_a,y_a,t)\in a$ and~$(x_b,y_b,t)\in b$ with~$x_a\leq x_b$,~$y_a\geq y_b$.
			If there were some matching~$\mu$ with~$a\in D$, then~$(x_a,y_\mu,t)\in\mu$ for some~$y_\mu>y_a$.
			Similarly, if~$b\notin D$, then~$(x_b,y'_\mu,t)\in\mu$ for some~$y'_\mu<y_b$.
			We can further deduce from~$x_a\leq x_b$ and monotonicity of~$\mu$ that we can pick~$y'_\mu$ such that~$y_a<y_\mu\leq y'_\mu<y_b$.
			However, this contradicts~$y_a\geq y_b$, so such a matching does not exist.
			Hence,~$a\in D \implies b\in D$ whenever~$\inR{a}{b}$ and therefore whenever~$\inRtrans{a}{b}$.
		\end{proof}
		\begin{lem}\label{lem:matchIfD}
			If~$d\in D\implies u\in D$, then~$\inRtrans{d}{u}$.
		\end{lem}
		\begin{proof}
			Suppose~$d\in D\implies u\in D$ but not~$\inRtrans{d}{u}$.
			Then no matching exists, and no path from~$d$ to~$u$ exists in the directed graph~$G=(O,\R)$.
			Pick as~$D$ the set of obstacles reachable from~$d$ in~$G$, then~$D$ does not contain~$u$.
			Pick the \emph{tightest} matching~$\mu$ such that~$D$ lies below it, we define~$\mu$ in terms of matchings~$\pi_t\subseteq \real^2\times\{t\}$ in the plane at each timestamp~$t$.
			$$(x,y,t)\in\pi_t \text{ \tIff~} (x'>x \land y'<y) \implies \neg m(x',y',t) \land m(x,y,t)\text{ where}$$
			$$m(x,y,t) \text{ \tIff~} \{(x',y',t)\mid x'\leq x\land y'\geq y \} \cap \bigcup D=\emptyset$$

			Because~$u\notin D$, this defines a monotone path~$\pi_t$ from~$(0,0)$ to~$(P,P)$ at each timestamp~$t$.
			Suppose that~$\pi_t$ \emph{properly} intersects some~$o\in O$, such that some point of~$(x_o,y_o,t)\in o$ lies below~$\pi_t$.
			It follows from the definition of~$\R$ and~$\neg m(x_o,y_o,t)$ that~$\inR{d}{o}$ for some~$d\in D$.
			However, such obstacle~$o$ cannot exist because~$D$ satisfies~$\R$.
			As a result, no path~$\pi_t$ intersects any obstacle and we can connect the paths~$\pi_t$ to obtain a continuous matching~$\mu$ without intersecting any obstacles.
			So~$\mu$ does not intersect obstacles in~$O\setminus D$, contradicting~$d\in D\implies u\in D$.
		\end{proof}
		\begin{thm}\label{thm:fdrel}
			The \FD is greater than~$\e$ \tIff~$\inRtrans{d}{u}$ for~$\e$.
		\end{thm}
		\begin{proof}
			We have for every matching that~$u\notin D$ and~$d\in D$.
			Therefore it follows from Lemma~\ref{lem:witnessIfNoD} that no matching exists if~$\inRtrans{d}{u}$ for~$\e$.
			In that case, the \FD is greater than~$\e$.
			Conversely, if~$\notinRtrans{d}{u}$ there is a set~$D$ satisfying~\Rtrans with~$u\notin D$ and~$d\in D$.
			In that case, a matching exists by Lemma~\ref{lem:matchIfD}, and the \FD is less than~$\e$.
		\end{proof}
		We choose the set of obstacles~$O'$ such that~$\bigcup O' = ([0,P]^2\times[0,T])\setminus \F_\e$ and the relation~$\R$ is easily computable.
		Note that due to Lemma~\ref{thm:convex}, each connected component contains a corner of a cell, so any cell in the freespace contains constantly many (up to eight) components of~$\bigcup O'$.
		Moreover, we can index each obstacle in~$O'$ by a grid point~$(x,y,t)\in\nat^3$.

		Let $o_{x,y,t}\subseteq ([0,P]^2\times[0,T])\cap([x-1,x+1]\times[y-1,y+1]\times[t-1,t+1])\setminus \F_\e$ be the maximal connected subset of the cells adjacent to $(x,y,t)$, such that $o_{x,y,t}$ contains $(x,y,t)$.
		Now, the obstacle~$o_{x,y,t}$ is not well-defined if $(x,y,t)\in\F_\e$, in which case we define~$o_{x,y,t}$ to be an empty (dummy) obstacle. We have~$O' = \bigcup_{(x,y,t)} \{o_{x,y,t}\}$ and we remark that obstacles are not necessarily disjoint.

		Each of the~$O(P^2T)$ obstacles is now defined by a constant number of vertices.
		We therefore assume that for each pair of obstacles~$(a,b)\in O^2$, we can decide in constant time whether~$\inR{a}{b}$; even though this decision procedure depends on the chosen distance metric.
		For each obstacle~$a=o_{x,y,t}$, there are~$O(P^2)$ obstacles~$b=o_{x',y',t'}$ for which~$\inR{a}{b}$, namely because~$t-2\leq t'\leq t+2$ if~$\inR{a}{b}$.
		Furthermore,~$u$ and~$d$ contribute to~$O(P^2T)$ elements of the relation.
		Therefore we can compute the relation~$\R$ in~$O(P^4T)$ time.

		Testing whether~$\inRtrans{d}{u}$ is equivalent to testing whether a path from~$d$ to~$u$ exists in the directed graph~$(O,\R)$, which can be decided using a depth first search.
		So we can solve the decision problem for the \FD in~$O(P^2T+|\R|)=O(P^4T)$ time.
		However, the relation~$\R$ may yield many unnecessary edges.
		In Section~\ref{sec:small} we show that a smaller set~$E$ of size~$O(P^3T)$ with the same transitive closure~$\Rtrans$ is computable in~$O(P^3T\log P)$ time, so the decision algorithm takes only~$O(P^3T\log P)$ time.

	\subsection{Parametric Search}\label{sec:paramDyn}
		\begin{wrapfigure}{r}{0pt}%
			\includegraphics{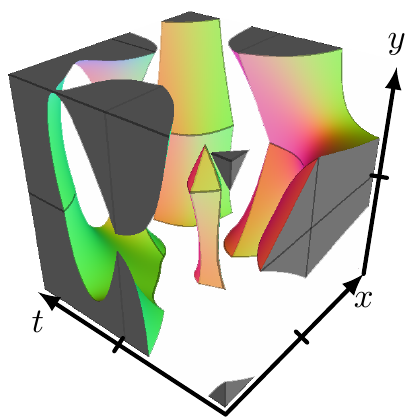}%
			\caption{$[0,2]^3\setminus \fsdIII$}%
			\label{fig:render}%
		\end{wrapfigure}

		To give an idea of what the 3D freespace looks like, we have drawn the obstacles of the eight cells of the freespace between two quadrilateral meshes of size~$P\times T=2\times 2$ in Figure~\ref{fig:render}.
		Cells of the 3D freespace lie within cubes, having six faces and twelve edges.
		We call such edges~$x$-,~$y$- or~$t$-edges, depending on the axis to which they are parallel.

		We are looking for the minimum value of~$\e$ for which a matching exists.
		When increasing the value of~$\e$, the relation~$\R$ becomes sparser since obstacles shrink.
		Critical values of~$\e$ occur when~$\R$ changes.
		Due to Lemma~\ref{thm:convex}, all critical values involve an edge or an~$xt$-face or~$yt$-face of a cell, but never the internal volume, so the following critical values cover all cases.

		\begin{enumerate}
			\item[a)]The minimal~$\e$ such that~$(0,0,t)\in\fsdIII$ and~$(P,P,t)\in\fsdIII$ for all~$t$.
			\item[b)]An edge of~$C_{x,y,t}$ becomes nonempty.
			\item[c)]
			Endpoints of~$y$-edges of~$C_{x,y,t}$ and~$C_{x+i,y,t}$ align in~$y$-coordinate, or
			endpoints of~$x$-edges of~$C_{x,y,t}$ and~$C_{x,y-j,t}$ align in~$x$-coordinate.
			\item[d)]Endpoints of a~$t$-edge of~$C_{x,y,t}$ and a~$t$-edge of~$C_{x+i,y-j,t}$ align in~$t$-coordinate.
			\item[e)]An obstacle in~$C_{x,y,t}$ stops overlapping with an obstacle in~$C_{x+i,y,t}$ or~$C_{x,y-j,t}$ when projected along the~$x$- or~$y$-axis.
		\end{enumerate}

		The endpoints involved in the critical values of type~a),~b),~c) and~d) can be captured in~$O(P^2T)$ functions.
		We apply a parametric search for the minimum critical value~$\e_\text{abcd}$ of type~a),~b),~c) or~d) for which a matching exists.
		This takes~$O((P^2T+\mtime_\text{dec})\log(PT))$ time.

		\begin{wrapfigure}{r}{0pt}%
			\includegraphics{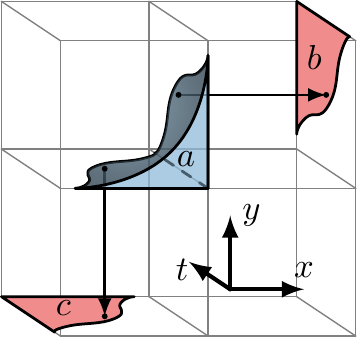}%
			\caption{$\inR{a}{b}$ and $\inR{a}{c}$}%
			\label{fig:overlap}%
		\end{wrapfigure}%
		We illustrate the need for critical values of type~e) in Figure~\ref{fig:overlap}, here obstacle~$a$ overlaps with both obstacles~$b$ and~$c$ while the overlap in edges does not contribute to~$\R$.
		It is unclear how critical values of type~e) can be incorporated in the parametric search directly.
		Instead, we enumerate and sort the~$O(P^3T)$ critical values of type~e) in~$O(P^3T\log(PT))$ time.
		Using~$O(\log(PT))$ calls to the decision algorithm, we apply a binary search to find the minimum critical value~$\e_\text{e}$ of type~e) for which a matching exists.
		Finding the critical value~$\e_\text{e}$ then takes~$O((P^3T+\mtime_\text{dec})\log(PT))$ time.
		The \matchID \FD is then the minimum of~$\e_\text{abcd}$ and~$\e_\text{e}$. This results in the following running time.
		\begin{thm}\label{thm:fdID}
		The \matchID \FD can be computed in $O((P^3T+\mtime_\text{dec})\log(PT))$ time.
		\end{thm}
Before stating the final running time, we present a faster algorithm for the decision algorithm.
\subsection{A Faster Decision Algorithm}\label{sec:small}
	To speed up the decision procedure we distinguish the cases for which two obstacles may be related by~\R, these cases correspond to the five types of critical values of~Section~\ref{sec:paramDyn}.
	Critical values of type~a) and~b) depend on obstacles in single cells, so there are at most~$O(P^2T)$ elements of~\R arising from type~a) and~b).
	Critical values of type~c) and~e) arise from pairs of obstacles in cells in the same row or column, so there are at most~$O(P^3T)$ of them.
	In fact, we can enumerate the edges of type~a),~b),~c), and~e) of~\R in~$O(P^3T)$ time.
	On the other hand, edges of type~d) arise between two cells with the same value of~$t$, so there can be~$O(P^4T)$ of them.

	We compute a smaller directed graph~$(V,E)$ with~$|E|=O(P^3T)$ that has a path from~$d$ to~$u$ if and only if~$\inRtrans{d}{u}$.
	Let~$V=O=\{u,d\}\cup O'$ be the vertices as before (we will include dummy obstacles for grid points in that lie in the freespace) and transfer the edges in~\R except those of type~d) to the smaller set of edges~$E$.
	We must still induce edges of type~d) in~$E$, but instead of adding~$O(P^4T)$ edges, we use only~$O(P^3T)$ edges.
	The edges of type~d) can actually be captured in the transitive closure of~$E$ using only~$O(P)$ edges per obstacle in~$E$.

	Using an edge from~$o_{x,y,t}$ to~$o_{x+1,y,t}$ and to~$o_{x,y-1,t}$, we construct a path from~$o_{x,y,t}$ to any obstacle~$o_{x+i,y-j,t}$.
	The sole purpose of the dummy obstacles is to construct these paths effectively.
	For obstacles whose gridpoints have the same~$t$-coordinates, it then takes a total of~$O(P^2T)$ edges to include the obstacles overlapping in~$t$-coordinate related by type~d), this is valid because $(x,y,t)\in o_{x,y,t}$ for non-dummy obstacles.

	Denote by~$\E{k}$ the edges of type~d) of the form~$(a,b)=(o_{x,y,t_a},o_{x+i,y-j,t_b})$ where~$t_b=t_a+k$, then the set~\E{0} of~$O(P^2T)$ edges is the one we just constructed.
	Now it remains to induce paths with~$t_a\neq t_b$, that still overlap in~$t$-coordinates, i.e. the sets~\E{-2},~\E{-1},~\E{1} and~\E{2}.
	Denote by~$t^-(a)$ and~$t^+(a)$ the minimum and maximum~$t$-coordinate over points in an obstacle~$a$.
	For each obstacle, both the~$t^-(a)$ and the~$t^+(a)$ coordinates are an endpoint of a $t$-edge in a cell defining the obstacle due to Lemma~\ref{thm:convex}, and thus computable in constant time.
	
	\begin{wrapfigure}{r}{0pt}
		\includegraphics{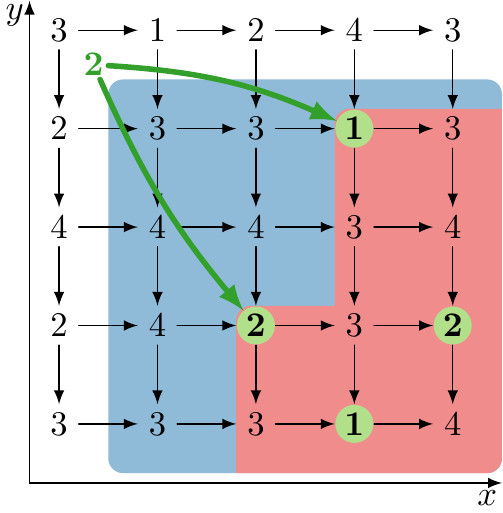}
		\caption{Two edges (green) cover (red) all four obstacles~$b$ (green) within the query rectangle (blue) with values~$t^-(b)\leq t^+(a)=2$.\label{fig:pareto}}
	\end{wrapfigure}
	Our savings arise from the fact that~\E{0} induces a path from~$o_{x+i,y-j,t+k}$ to~$o_{x+i',y-j',t+k}$ if $\inR{o_{x,y,t}}{o_{x+i,y-j,t+k}}$ and $\inR{o_{x,y,t}}{o_{x+i',y-j',t+k}}$ with~$i\leq i'$ and~$j\leq j'$, so we do not need an additional edge to induce a path to the latter obstacle.
	To avoid degenerate cases, we start by exhaustively enumerating edges of~\E{k} ($k\in\{-2,-1,1,2\})$ for which~$i\leq 1$ or~$j\leq 1$ in~$O(P^3T)$ time so we need only consider edges with~$i\geq 2\land j\geq 2$.

	For these remaining cases, we have~$\inR{a}{b}$ \tIff~$t^+(a)\geq t^-(b)\land t_b=t_a+k$, and~$t^-(a)\leq t^+(b)\land t_b=t_a-k$ for positive~$k$.
	From this we can derive the edges of~\E{k}.
	Although for each~$a$, there may be~$O(P^2)$ obstacles~$b$ such that~$\inR{a}{b}$ with~$t_b=t_a+k$, the Pareto frontier of those obstacles~$b$ contains only~$O(P)$ obstacles, see the grid of fictional values~$t^-(b)$ in~Figure~\ref{fig:pareto}.
	In Section~\ref{sec:pareto}, we show how to find these Pareto frontiers in~$O(P\log P)$ time per obstacle~$a$, using only~$O(P^2T)$ preprocessing time for the complete freespace.

	As a result, we can compute all~$O(P^3T)$ edges of~\E{k} in~$O(P^3T\log P)$ time.
	By Theorem~\ref{thm:small}, the decision problem for the \matchID \FD is solvable in~$O(P^3T\log P)$ time.

	\newcommand{\thmSmall}{The decision problem for the \matchID~\FD is solvable in~$O(P^3T\log P)$ time.}
	\begin{thm}\label{thm:small}\thmSmall\end{thm}
	\begin{proof}
	The edges~$E$ of types other than~d) are enumerated in~$O(P^3T)$ time, and using constantly many Pareto frontier queries for each obstacle,~$O(P^3T)$ edges of type~d) in~$E$ are computed in~$O(P^3T\log P)$ time.
	Given the set~$E$ of edges, deciding whether a path between two vertices exists takes~$O(|E|)=O(P^3T)$ time.
	The transitive closure of~$E$ equals~$\Rtrans$, so a path from~$d$ to~$u$ exists in~$E$ \tIff there was such a path in~$\R$.
	Since we compute~$E$ in~$O(P^3T\log P)$ time, the decision problem is solved in~$O(P^3T\log P)$ time.
	\end{proof}
The following immediately follows from Theorems~\ref{thm:fdID} and~\ref{thm:small}.
\begin{cor}\label{cor:runningtime}
The \matchID \FD can be computed in time\linebreak $O(P^3T\log P\log(PT))$.
\end{cor}

	\subsection{Pareto Frontier Queries}\label{sec:pareto}
		Suppose we are given a matrix~$M$ with~$m$ columns and~$n$ rows, and we want to be efficiently query submatrices for the Pareto frontier of numbers that are at most a given threshold value,~$t$.
		A query specifies the threshold~$t$, and two coordinates~$(x_0,y_0)$ and~$(x_1,y_1)$ of the query rectangle~$R=\{x_0,\dots,x_1\}\times\{y_0,\dots,y_1\}$.
		Let~$C_t(R)=\{(x,y)\in R\mid M[x,y]\leq t\}$ denote the coordinates of the query rectangle that must be dominated by the Pareto frontier~$F_t(R)\subseteq C_t(R)$.
		That is, if~$(x,y)\in C_t(R)$, then some~$(x',y')\in F_t(R)$ with~$x'\leq x \land y'\leq y$ exists.

		We preprocess each of the~$O(n)$ rows of~$M$ in~$O(m)$ time by storing their cells as the leaves of an augmented binary tree, whose internal nodes store the minimum value over its subtrees. 
		Then queries for the index of the leftmost element with value at most~$t$ in a range~$\{x_0,\dots,x_1\}$ of that row can be answered in~$O(\log m)$ time.
		We can compute~$F_t$ by including for each row, the element with minimum~$x$-coordinate and value at most~$t$ in the query range.
		So using~$O(n)$ queries, we compute a set~$F_t(R)$ of size~$O(n)$ in~$O(n\log m)$ time, using~$O(nm)$ preprocessing time for the matrix~$M$.

		In this case,~$F_t(R)$ is not actually the Pareto frontier since some of its elements might be dominated by other elements.
		With a slight modification, we can make~$F_t(R)$ the actual Pareto frontier (of minimum size).
		If the query with range~$\{x_0,\dots,x_1\}$ returns~$x^*$ with~$x_0\leq x^*\leq x_1$ for some row, then the query range for subsequent rows can be restricted to~$\{x_0,\dots,x^*-1\}$, so no unnecessary values are generated.

	\section{Hardness}\label{sec:hardness}
		\begin{wrapfigure}{l}{0pt}\homCC\hspace{1em}\homCD\end{wrapfigure}
		We extend the \matchIC and \matchID classes of matchings (of Sections~\ref{sec:constant} and~\ref{sec:dynamic}) to asynchronous ones.
		For this, we allow realignments of timestamps, giving rise to the \matchCC and \matchCD classes of matchings.
		The \matchCC class ranges over matchings of the form~$\mu(p,t)=(\pi(p),\tau(t))$ where the~$\pi$ and~$\tau$ are matchings of positions and timestamps.
		The \matchCD class of matchings has the form~$\mu(p,t)=(\pi_t(p),\tau(t))$ for which the positional matching~$\pi_t$ changes over time.
		We first prove that the \matchCC~\FD is in~NP.
		\begin{thm}\label{thm:NP}
			Computing the \FD is in NP for the \matchCC class of matchings.
		\end{thm}
		\begin{proof}
			Given any matching~$\mu(p,t)=(\pi(p),\tau(t))$ with a \FD of~$\e$, we can derive---due to Lemma~\ref{thm:convex}---a piecewise-linear matching~$\tau^*$ in~$O(T)$ time, such that a matching~$\mu^*(p,t)=(\pi^*(p),\tau^*(t))$ with \FD at most~$\e$ exists.
			We can realign the quadrilateral meshes~$A$ and~$B$ under~$\tau^*$ to obtain meshes~$A^*$ and~$B^*$ of polynomial size.
			Now the polynomial-time decision algorithm for \matchIC matchings (see full paper) is applicable to~$A^*$ and~$B^*$.
		\end{proof}
		Due to critical values of type e), it is unclear whether each \matchCD matching admits a piecewise-linear matching~$\tau^*$ of polynomial size, which would mean that the \matchCD~\FD is also in~NP.

		We show that computing the \FD is NP-hard for both classes by a reduction from~3-SAT.
		The idea behind the construction is illustrated in the two height maps of Figure~\ref{fig:terrainsNPC}.
		These represent quadrilateral meshes embedded in~$\real^1$ and correspond to a single clause of a~3-CNF formula of four variables.

		\begin{figure}[ht!]\centering%
			\begin{minipage}{0.85\columnwidth}
				\includegraphics[scale=.9]{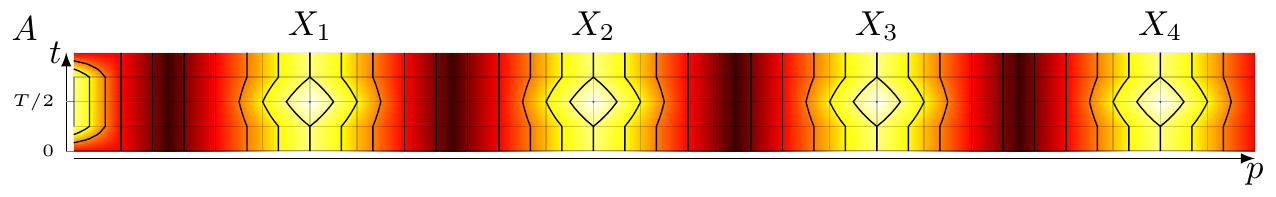}\\
				\includegraphics[scale=.9]{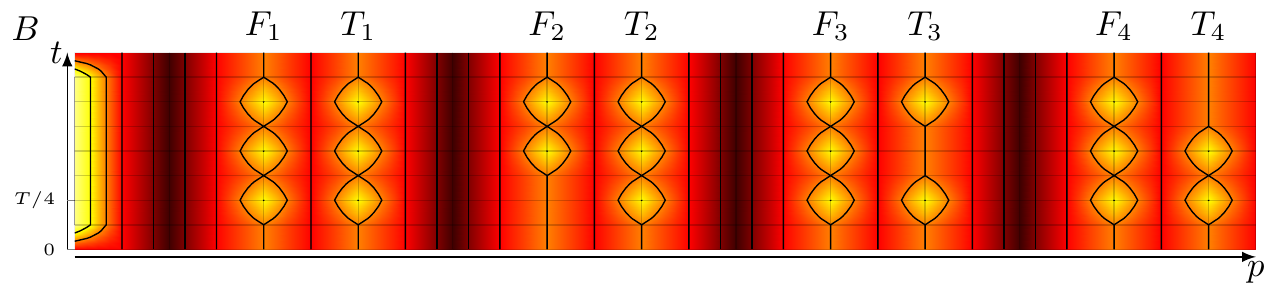}
			\end{minipage}%
			\begin{minipage}{0.05\columnwidth}
				\includegraphics[scale=.9]{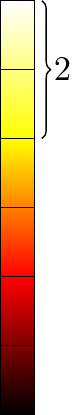}
			\end{minipage}
			\caption{Two quadrilateral meshes~$A$ and~$B$ embedded in~$\real^1$ (indicated by color and isolines). Their \FD is~$2$ isolines if the clause~$(X_2\lor\neg X_3\lor\neg X_4)$ is satisfiable and~$3$ isolines otherwise.\label{fig:terrainsNPC}}
		\end{figure}
		\begin{figure}[ht!]\centering%
			\includegraphics[width=.4\textwidth]{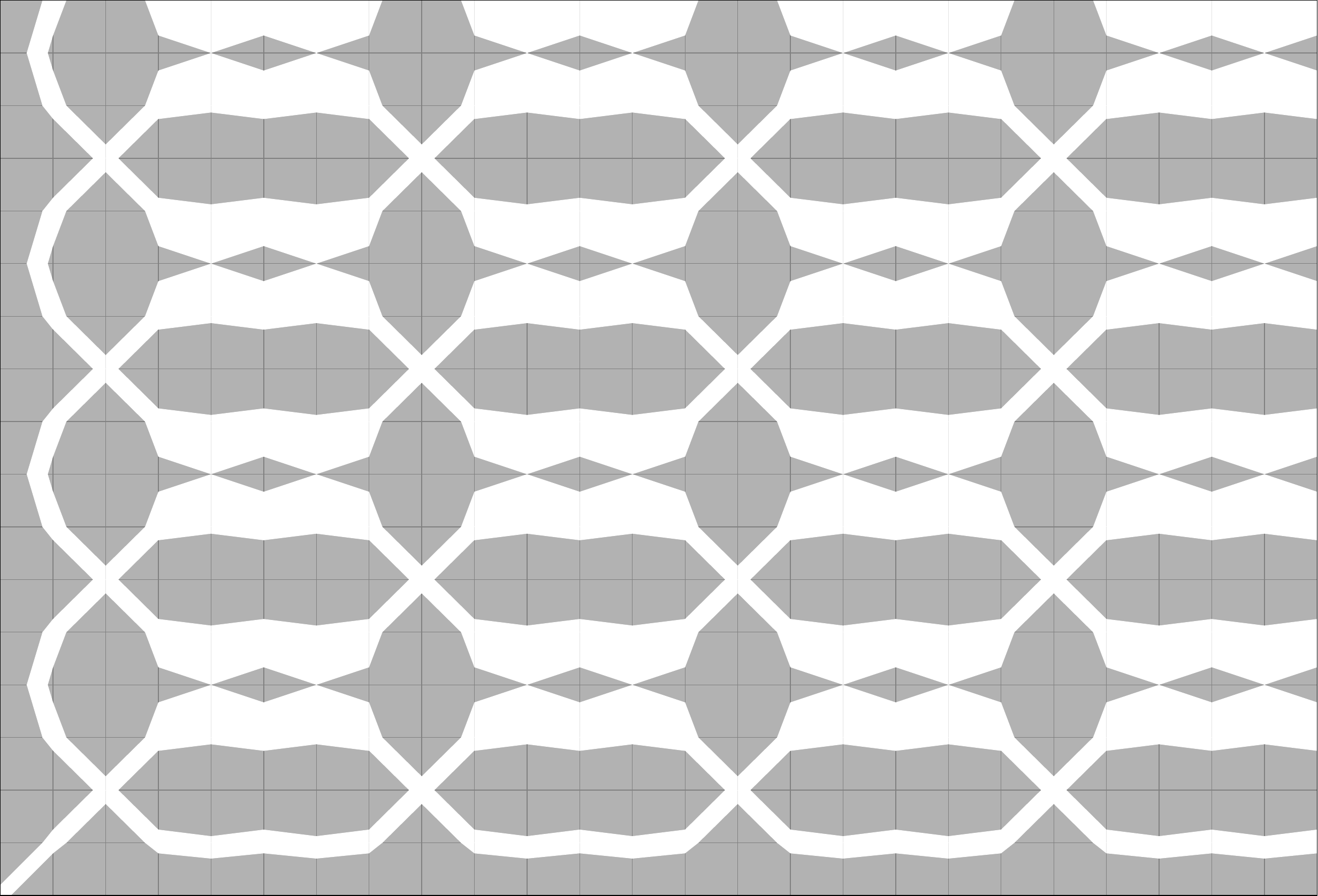}\qquad%
			\includegraphics[width=.4\textwidth]{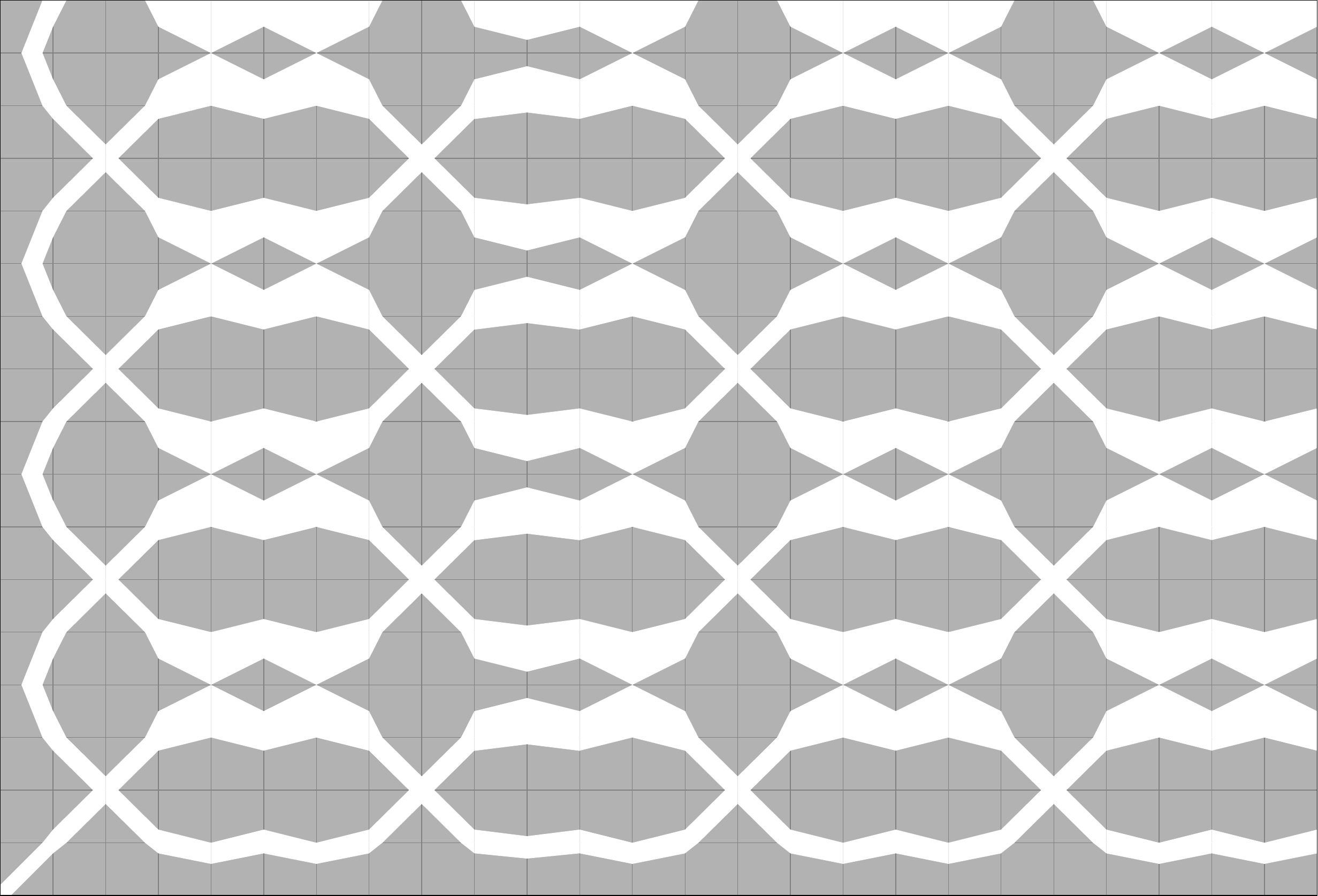}%
			\caption{The freespace~$\F_{\e=2}$ of $(A,B)$ at times $(0,0)$ (left) and $(T/2,T/4)$ (right).\label{fig:FSDterrainsNPC}}
		\end{figure}

		We distinguish valleys (dark), peaks (white on~$A$, yellow on~$B$) and ridges (denoted~$X_i$,~$F_i$ and~$T_i$).
		An important observation is that in order to obtain a low \FD of~$\e<3$, the $n$-th valley of~$A$ must be matched with the~$n$-th valley of~$B$.
		Moreover, each ridge~$X_i$ must be matched with~$F_i$ or~$T_i$ and each peak of~$A$ must be matched to a peak of~$B$.
		Note that even for \matchCD matchings, if~$X_i$ is matched to~$F_i$, it cannot be matched to~$T_i$ and vice-versa because the (red) valley separating~$F_i$ and~$T_i$ has distance~$3$ from~$X_i$.

		The aforementioned properties are reflected more clearly in the 2D freespace between the curves at aligned timestamps~$t$ and~$\tau(t)$.
		In Figure~\ref{fig:FSDterrainsNPC}, we give two 2D slices (with $(t_A,t_B)=(0,0)$ and~$(T/2,T/4)$, respectively) of the 4D-freespace diagram with~$\e=2$ for the shown quadrilateral meshes.
		In this diagram with~$\e=2$, only~$2^3$ monotone paths exist (up to directed homotopy) whereas for~$\e=3$ there would be~$2^4$ monotone paths (one for each assignment of variables).
		For~$\e=2$, the peak of~$X_2$ cannot be matched to~$F_2$ at~$t=T/4$ of~$B$, corresponding to an assignment of~$X_2=\textit{true}$.

		Consider a 3-CNF formula with~$n$ variables and~$m$ clauses, then~$A$ and~$B$ consist of~$m$ clauses along the~$t$-axis and~$n$ variables ($X_1\dots X_n$ and~$F_1,T_1\dots F_n,T_n$) along the~$p$-axis.
		The~$k$-th clause of~$A$ is matched to the~$k$-th clause of~$B$ due to the elevation pattern on the far left~($p=0$).
		This means that the peaks of~$A$ are matched with peaks of the same clause on~$B$ and all these peaks have the same timestamp because~$\tau(t)$ is constant (independent of~$p$).

		For each clause, there are three rows (timestamps) of~$B$ with peaks on the ridges.
		On each such timestamp, exactly one ridge (depending on the disjuncts of the clause) does not have a peak.
		Specifically, if a clause has~$X_i$ or~$\neg X_i$ as its~$k$-th disjunct, then the~$k$-th row of that clause has no peak on ridge~$F_i$ or~$T_i$, respectively.
		We use these properties in Theorem~\ref{thm:fdApx} where we prove that it is NP-hard to approximate the~\FD within a factor~$1.5$.

		\begin{lem}\label{thm:fd3}
		The~\FD between two such moving curves is at least~$3$ if the corresponding~3-CNF formula is unsatisfiable.
		\end{lem}
		\begin{proof}
		Consider a matching yielding a~\FD smaller than~$3$ given an unsatisfiable formula, then the peaks of~$A$ (of the~$k$-th clause) are matched with peaks of~$B$~(of a single row of the~$k$-th clause).
		Assign the value \textit{true} to variable~$X_i$ if ridge~$X_i$ is matched with~$T_i$ and \textit{false} if it is matched with~$F_i$.
		Then for every clause~$(V_i\lor V_j\lor V_k)$ with~$V_i\in\{X_i,\neg X_i\}$, there is a peak at~$\pi(X_i)$,~$\pi(X_j)$ or~$\pi(X_k)$ for that clause.
		Such a matching cannot exist because then the~3-CNF formula would be satisfiable, so the~\FD is at least~$3$.
		\end{proof}
		\begin{lem}\label{thm:fd2}
		The~\FD between two such moving curves is at most~$2$ if the corresponding~3-CNF formula is satisfiable.
		\end{lem}
		\begin{proof}
		Consider a satisfying assignment to the 3-CNF formula.
		Match~$X_i$ with the center of~$F_i$ or~$T_i$, if~$X_i$ is \textit{false} or~\textit{true}, respectively.
		For every clause, the timestamp with peaks of~$A$ can be matched with a row of peaks on~$B$.
		As was already hinted at by~Figure~\ref{fig:terrainsNPC}, the remaining parts of the curves can be matched with~$\e=2$.
		Therefore this yields a \FD of at most~$2$.
		\end{proof}
		\begin{thm}\label{thm:fdApx}
		It is NP-hard to approximate the~\matchCC or~\matchCD~\FD for moving curves in~$\real^1$ within a factor~$1.5$.
		\end{thm}
		\begin{proof}
		By Lemmas~\ref{thm:fd3} and~\ref{thm:fd2}, the~\matchCC or~\matchCD~\FD between two quadrilateral meshes embedded in~$\real^1$ is at least~$3$ or at most~$2$, depending on whether a 3-CNF formula is satisfiable.
		\end{proof}

	\subsection{Orientation-Preserving Homeomorphisms}\label{sec:orient}

		\begin{wrapfigure}{l}{0pt}\homDD\end{wrapfigure}
		Previous results~\cite{surface,simpleHard} have shown that computing the \FD between surfaces under orientation-preserving homeomorphisms is NP-hard for surfaces embedded in~$\real^2$.
		We will refer to this variant as the~\emph{\matchDD}~\FD.
		The prior results hold for triangular meshes, which are a degenerate case of quadrilateral meshes.
		Although the \matchDD \FD is upper-semicomputable, it is unknown whether it is computable at all.
		We consider the case where the corners of the meshes are matched with each other.

		The prior NP-hardness constructions for the~\matchDD~\FD seem unnecessarily complex.
		Therefore we extend our results for the \matchCD~\FD of Section~\ref{sec:hardness} to obtain a new hardness construction for the \matchDD~\FD for surfaces embedded in~$\real^2$.
		Note that we cannot directly apply the previous construction because there was only a single matching of timestamps for~\matchCD matchings.
		Thus, we do not preserve the property that for a clause, the timestamp with peaks of~$A$ maps to a single timestamp of~$B$.

		We prevent this by means of a second dimension in which we enforce that a row of peaks of~$A$ maps to a single row of peaks of~$B$.
		In addition to the embedding in~$\real^1$ of Figure~\ref{fig:terrainsNPC}, which defines one coordinate for every point on a clause of a quadrilateral mesh, we define a second coordinate using Figure~\ref{fig:terrainsNPC2}, yielding an embedding in~$\real^2$.

		\begin{figure}[h!]\centering%
			\begin{minipage}{0.85\textwidth}
				\includegraphics[scale=.9]{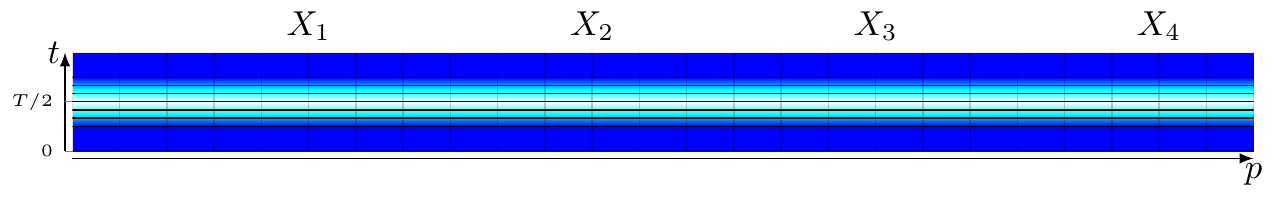}\\
				\includegraphics[scale=.9]{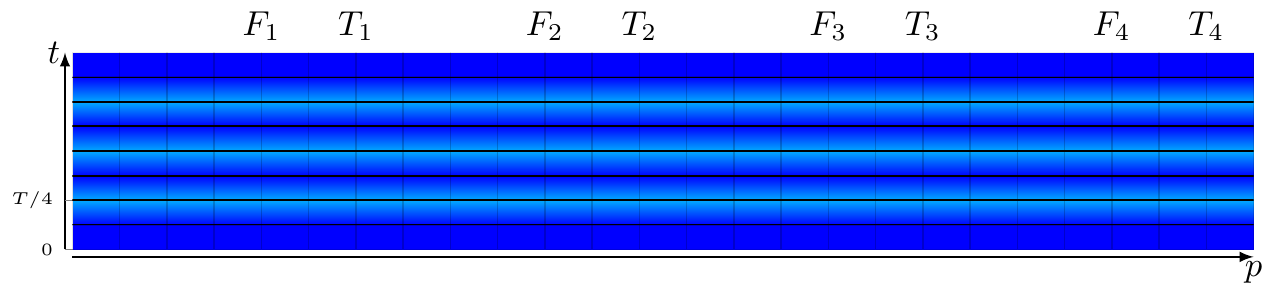}
			\end{minipage}%
			\begin{minipage}{0.05\textwidth}
				\includegraphics[scale=.9]{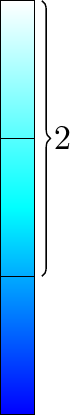}
			\end{minipage}
			\caption{The second coordinate of a clause gadget of~$A$ (top) and~$B$ (bottom) embedded in~$\real^2$, the first coordinate is given in Figure~\ref{fig:terrainsNPC}.\label{fig:terrainsNPC2}}
		\end{figure}

		Now, under the maximum norm ($L^\infty$), a row of peaks of~$A$ can only be matched to a single row of peaks on~$B$ for~$\e<3$.
		Conversely, we can still match the meshes of two satisfiable formulas with~$\e=2$.
		Hence, Theorem~\ref{thm:fdOrient} follows.
		This result extends to triangular meshes since all quadrilaterals lie in the plane, and can thus be represented by a pair of triangles.
		For norms other than the maximum norm, the problem is still NP-hard, but our bound on the approximation factor is smaller than~$1.5$.

		\begin{thm}\label{thm:fdOrient}
		Unless~P=NP, no polynomial time algorithm can approximate the \matchDD \FD between two quadrilateral meshes embedded in~$\real^2$ under the maximum norm within a factor~$1.5$.
		\end{thm}

	\section{Conclusion}

		Based on the~\FD, we presented several similarity measures between moving curves, together with efficient algorithms for computing some measures, while proving NP-hardness for computing others.

		Although many algorithmic solutions to variants of the \FD between static curves also apply to the \matchIC~\FD, extending the \matchID class is more complex.
		For example, computing the \matchID~\FD between moving closed curves (with a cylinder instead of square as parameter space) remains an open problem.

		Finally, the performance of the algorithm for the \matchID~\FD seems undesirably slow and we have not been able to prove tight lower bounds on this running time yet.

\mypar{Acknowledgements} K. Buchin, T. Ophelders, and B. Speckmann are supported by the Netherlands Organisation for Scientific Research (NWO) under project no.
612.001.207 (K. Buchin) and no. 639.023.208 (T. Ophelders \& B. Speckmann).

\bibliographystyle{abbrv}
\bibliography{ref}
\end{document}